\documentclass[11pt]{article}

\usepackage{amsmath,amssymb,amsthm,epsfig,color,bbm,hyperref,verbatim,vmargin}
\hypersetup{
	colorlinks=true,
	linkcolor=blue,
	citecolor=blue,
	urlcolor=blue
}
\newtheorem{theorem}{Theorem}

\newtheorem{remark}[theorem]{Remark}
\newtheorem{corollary}[theorem]{Corollary}
\newtheorem{lemma}[theorem]{Lemma}
\newtheorem{proposition}[theorem]{Proposition}

\numberwithin{equation}{section}
\numberwithin{theorem}{section}

\title{Moments at the hard edge and Rayleigh functions}

\author{
  Anna Maltsev\thanks{School of Mathematical Sciences, Queen Mary Univesity of London, Mile End Road, London, E1 4NS, UK. Email: \texttt{a.maltsev@qmul.ac.uk}} \quad \text{and} \quad
  Nick Simm\thanks{Department of Mathematics, University of Sussex, Falmer, Brighton, BN1 9RH UK. Email: \texttt{n.j.simm@sussex.ac.uk}}
}

\date{}
\date{}
\begin{document}
\maketitle
\begin{abstract}
Motivated by the analogy between spectral moments of random matrices and associated zeta functions, we study inverse power trace moments of the Laguerre ensemble of dimension $N$ and inverse temperature parameter $\beta>0$. We consider a large $N$ regime determined by the low-lying eigenvalues of the ensemble known as the hard edge. In the classical cases $\beta \in \{1,2,4\}$, we obtain explicit results for the inverse moments and extend these to formulae for the corresponding Mellin transforms. In the case of general $\beta>0$, by a result of Fyodorov and Le Doussal \cite{FD16}, we obtain a different formula for the moments given as a sum over partitions. We use this to consider a low temperature limit where $\beta \to \infty$ as $N \to \infty$. In this limit, we show that the moments are given in terms of the Bessel zeta function. 
\end{abstract}

\section{Introduction and main results}
The topic of spectral moments has a long history in random matrix theory, dating back to Wigner's work on the semi-circle law in 1955. Since then, several connections to other areas have been found, such as the use of spectral moments for computing the Euler characteristic of moduli spaces \cite{HZ86}, as well as the connection to genus expansion and orthogonal polynomials \cite{HT03, L09}. There is continued interest in this topic from various perspectives. For example, recent works treat non-Hermitian and elliptic ensembles, leading to explicit formulae for mixed moments and Harer--Zagier type recursions for real-eigenvalue moments \cite{ABO25,B24}. Another direction concerns deformed ensembles: $q$-deformed Gaussian and Laguerre models exhibit large-$N$ genus expansions and a $q$-deformation of the Marchenko--Pastur law 
\cite{BJM26,B26}. Spectral moments also arise in combinatorics and applications through ribbon graph descriptions of correction terms \cite{RGM26} and through Bures--Hall calculations related to entanglement entropy \cite{W26}.

This paper is inspired by the developments of \cite{CMOS19} relating spectral moments of random matrix ensembles to associated zeta functions. To introduce the problem we consider here, recall that the \textit{Laguerre ensemble} is the joint probability density function of eigenvalues proportional to
\begin{equation}
\prod_{j=1}^{N}\lambda_{j}^{\alpha}e^{-\lambda_{j}}\prod_{1 \leq i < j \leq N}|\lambda_{j}-\lambda_{i}|^{\beta}, \label{laguerrebeta}
\end{equation}
supported on $\lambda_{j} \in [0,\infty)$ for each $j=1,\ldots,N$. Here, $\beta>0$ is the inverse temperature parameter. In the special cases $\beta \in \{1,2,4\}$, this density arises from the eigenvalues of $XX^{\dagger}$ where $X$ is an $N \times m$ matrix of independent Gaussian random variables over the real ($\beta=1$), complex ($\beta=2$) or quaternion ($\beta=4$) numbers and $\alpha = \frac{\beta}{2}(m-N+1)-1$. Such matrices are known as Wishart matrices after their introduction in the multivariate statistics literature by Wishart in 1928. More recently, it was discovered that \eqref{laguerrebeta} can be realized in terms of the eigenvalues of certain tri-diagonal matrix models for any $\beta>0$ \cite{DE02}. In this setting, \eqref{laguerrebeta} can be interpreted as a statistical mechanics model of a Coulomb gas with inverse temperature $\beta$ and restricted to the positive half-line.

The subject of the present study is the following class of inverse moments
\begin{equation}
\mathcal{M}^{(\beta)}_{N}(-k,\alpha) = \mathbb{E}\left(\sum_{j=1}^{N}\frac{1}{\lambda_{j}^{k}}\right), \qquad k>0,\label{momsintro}
\end{equation}
where the expectation is taken with respect to \eqref{laguerrebeta}. One of our motivations for considering such moments pertains to the interpretation of \eqref{momsintro} as a spectral zeta function. This perspective was considered in \cite{CMOS19} where for $\beta=2$ it was shown that \eqref{momsintro} has the following remarkable properties.
\begin{theorem}[Theorem 4.1 in \cite{CMOS19}]
\label{thm:cmos}
Let $\beta=2$ and define $\xi_{N}(s) = \frac{\mathcal{M}^{(2)}_{N}(-s,\alpha)}{\Gamma(1+\alpha-s)}$. Then $\xi_{N}(s)$ can be analytically continued to an entire function of $s \in \mathbb{C}$ satisfying the functional equation
\begin{equation}
\label{reflection}
\xi_{N}(s) = \xi_{N}(1-s), \qquad s \in \mathbb{C}.
\end{equation}
Furthermore, all zeros of $\xi_{N}(s)$ lie on the vertical line $\mathrm{Re}(s) = \frac{1}{2}$.
\end{theorem}
In this paper we study the large $N$ asymptotics of \eqref{momsintro} with general $\beta>0$ and we keep the parameter $\alpha$ fixed as $N \to \infty$. This is known as a hard edge scaling regime of the eigenvalue point process \cite{Forr}. In this setting, the moments \eqref{momsintro} have a non-trivial dependence on the parameters $\alpha$ and $\beta$. Our aim here is to give the first systematic study of \eqref{momsintro} in the hard edge scaling regime: in particular, to identify the large $N$ asymptotic scaling and to determine the dependence on $\alpha$ and $\beta$ (Theorems \ref{thm:bet2}, \ref{thmbet1bet4}, Corollary \ref{thm:momsbeta}). Our second main contribution is to identify an appropriate large $N$ and large $\beta$ limit such that \eqref{momsintro} is given in terms of the Bessel zeta function (Theorem \ref{thm:lowtemp}). 

The Bessel zeta function is the zeta function associated with the positive zeros $j_{\nu,n}$ of the Bessel function $J_\nu$,
\begin{equation}
\zeta_\nu(s)=\sum_{n=1}^{\infty}\frac{1}{j_{\nu,n}^{s}}, \qquad \mathrm{Re}(s)>1, \label{besselzeta}
\end{equation}
with parameter $\nu>-1$. The even integer values $\zeta_{\nu}(2k)$ are known as \textit{Rayleigh functions}, we discuss these further in Section \ref{se:betaresults}. The special case $\nu=\tfrac12$ reduces to $\pi^{-s}\zeta(s)$ for the Riemann $\zeta$-function. Furthermore, for the one-dimensional Sturm–Liouville operator
\[
S_\nu=-\frac{d^2}{dx^2}+\frac{4\nu^2-1}{4x^2}
\]
with Dirichlet boundary condition at $x=l$, the eigenvalues are $j_{\nu,n}^2/l^2$, so its spectral zeta function is exactly $l^{2s}\zeta_\nu(2s)$ \cite{spreafico2004non}. The same mechanism underlies Laplacians on discs, cones, and circular sectors, where the radial eigenvalues are again determined by squares of Bessel zeros \cite{spreafico2005zeta}. Thus the Bessel zeta functions link spectral geometry and mathematical physics and they arise naturally in problems with cylindrical or spherical symmetry (see \cite{naber2022bessel,AB96} and references therein). In the present hard-edge setting they are the natural limiting objects, because the appropriately scaled $\beta$-Laguerre eigenvalues converge to Bessel zeros in a large $\beta$ limit.

In the cases $\beta \in \{1,2,4\}$, several results in the literature give exact formulae for \eqref{momsintro} of varying degrees of complexity \cite{LM04,M12,CSN14,MS11}. From these, in principle, it is possible to perform a large $N$ analysis, but this has only been done when $\alpha$ is scaled in proportion to $N$. In particular, \eqref{momsintro} was studied in the context of Wigner delay times arising in mesoscopic physics \cite{MS12,MS13,CMSV16,N22} and in the context of $1/N$ expansions \cite{CDO21,GGR20}. See also \cite{CCO20} for an extension to discrete orthogonal polynomial ensembles and \cite{ABGS21} in the Cauchy ensemble. 

A general way to gain insight into quantities of the form \eqref{momsintro} is to consider the limiting density of eigenvalues distributed according to \eqref{laguerrebeta}, which is the well-known Marchenko-Pastur law. To state this, let $m := m_N \to \infty$ such that $\frac{N}{m} \to c \in [0,1]$ and consider the scaled eigenvalues $x_{j} = 2\lambda_{j}/(\beta m)$. Then for every bounded and continuous function $f$ and any fixed $\beta>0$, we have (see \textit{e.g.}\ \cite[Eq. (4)]{DE06})
\begin{equation}
\lim_{N \to \infty}\,\frac{1}{N}\mathbb{E}\left(\sum_{j=1}^{N}f(x_j)\right) = \int_{c_{-}}^{c_{+}}f(x)\rho_{\mathrm{MP}}(x)\,dx, \label{mplaw}
\end{equation}
where $\rho_{\mathrm{MP}}$ is the Marchenko-Pastur density,
\begin{equation}
\rho_{\mathrm{MP}}(x) = \frac{\sqrt{(x-c_{-})(c_{+}-x)}}{2\pi c x}
\end{equation}
supported on $(c_{-},c_{+})$, where $c_{\pm} = (1\pm \sqrt{c})^{2}$. Note that when $c=1$, the lower end point of the support is $c_{-}=0$, which is the case if $\alpha$ in \eqref{laguerrebeta} is fixed and in the standard terminology, this is referred to as a hard edge \cite{Forr}. In particular, the right-hand side of \eqref{mplaw} with $f(x) = x^{-k}$ diverges for any $k \geq 1$ which shows that \eqref{mplaw} cannot be (directly) used to calculate the inverse moments in the presence of a hard edge.

\subsection{Main results for {\texorpdfstring{$\beta \in \{1,2,4\}$}{beta = 1,2,4}}}
Throughout the paper, the parameter $\alpha$ is kept fixed as $N \to \infty$, so that we are working in the hard edge regime of \eqref{mplaw}. We will show that for any $s \in \mathbb{C}$ in the strip $\mathcal{S} = \{\frac{1}{2} < \mathrm{Re}(s) < \alpha+1\}$ and $\beta \in \{1,2,4\}$, the following limit exists
\begin{equation}
	\mathcal{M}^{(\beta)}(-s,\alpha) := \lim_{N \to \infty}N^{-s}\mathcal{M}^{(\beta)}_{N}(-s,\alpha). \label{limit-exists}
\end{equation}
As we are working in the analyticity domain $\mathcal{S}$, our results are only meaningful when $\alpha > -\frac{1}{2}$. Our objective is to characterize the limiting quantity $\mathcal{M}^{(\beta)}(-s,\alpha)$.
\begin{theorem}[Hard edge Mellin transform for $\beta=2$]
\label{thm:bet2}
The limit \eqref{limit-exists} exists for $\beta=2$, $\alpha > - \frac{1}{2}$ and for any $s \in \mathcal{S}$. We have
\begin{equation}
\mathcal{M}^{(2)}(-s,\alpha) = \frac{4^{s-1}\Gamma(\alpha+1-s)\Gamma\left(s-\frac{1}{2}\right)}{\sqrt{\pi}\Gamma(s+1)\Gamma(s+\alpha)}. \label{bet2form}
\end{equation}
\end{theorem}
\begin{proof}
See Section \ref{se:proofsbet124}.
\end{proof}
We give two proofs of Theorem \ref{thm:bet2}, both of which are based on the fact that for $\beta=2$, \eqref{laguerrebeta} can be characterized in terms of a determinantal point process with an explicit kernel constructed in terms of Laguerre polynomials. The exact formulae that result were used in \cite{CMOS19} to establish Theorem \ref{thm:cmos}. Combining \eqref{reflection} and the Marchenko-Pastur result \eqref{mplaw} allows for a quick proof of Theorem \ref{thm:bet2}. The other proof we give is based on a more explicit computation in terms of the eigenvalue density corresponding to \eqref{laguerrebeta}. The advantage of the latter approach is that it extends to $\beta \neq 2$ and highlights the significance of the hard edge scaling in terms of Bessel functions \cite{F93,TW94}.

The cases $\beta=1$ and $\beta=4$ are known to be more challenging than $\beta=2$. In particular, the reflection property in \eqref{reflection} does not hold. In these cases \eqref{laguerrebeta} is governed by a Pfaffian point process and one often needs to work with skew-orthogonal polynomials. Several authors have worked to simplify this approach and reduce the problem to the $\beta=2$ case up to certain correction terms \cite{TW96,AFNvM00,W99}. Hence, the difficulty in the $\beta \in \{1,4\}$ cases is to effectively deal with such additional terms. 

In order to state our results, recall the definition of the hypergeometric series,
\begin{equation}
_pF_q(a_1,\ldots,a_p;b_1,\ldots,b_q;z) := \sum_{n=0}^{\infty}\frac{(a_1)_{n}\ldots(a_p)_{n}}{(b_1)_{n}\ldots(b_q)_{n}}\frac{z^{n}}{n!}.
\end{equation}
\begin{theorem}[Hard edge Mellin transform for $\beta=1$ and $\beta=4$]
	\label{thmbet1bet4}
The limit \eqref{limit-exists} exists for $\beta \in \{1,4\}$, $\alpha > -\frac{1}{2}$ and for any $s \in \mathcal{S}$. We have
\begin{equation}
\begin{split}
&\mathcal{M}^{(4)}(-s,\alpha) = 2^{s-1}\mathcal{M}^{(2)}(-s,\alpha)\\
&-\frac{2^{s-1}\Gamma(\alpha+1-s)}{\Gamma(\alpha+3)\Gamma(s-1)}{}_3F_2\left(\frac{\alpha}{2}+1,-s+\alpha+1,-s+2;\frac{\alpha}{2}+2,\alpha+2; 1\right) \label{hesymp}
\end{split}
\end{equation}
and
\begin{equation}
\begin{split}
&\mathcal{M}^{(1)}(-s,\alpha)  = 2^{s}\mathcal{M}^{(2)}(-s,2\alpha)+2^{s-1}\frac{\Gamma(\alpha+1-s)}{\Gamma(\alpha+1+s)}\\
&-2^{s}\frac{\Gamma(2\alpha+1-s)}{\Gamma(2\alpha+1)\Gamma(s+1)}{}_3F_2\left(\alpha,2\alpha+1-s,-s;\alpha+1,2\alpha; 1\right). \label{heortho}
\end{split}
\end{equation}
The formula \eqref{hesymp} can also be expressed as
\begin{multline}
\mathcal{M}^{(4)}(-s,\alpha)
=\frac{2^{s-1}(\alpha+3-s)\Gamma(\alpha+1-s)}
{\Gamma(\alpha+3)\Gamma(s)}\\
\times {}_4F_3\left(
\frac{\alpha+5-s}{2},\frac{\alpha}{2}+1,\alpha+1-s,2-s;
\frac{\alpha+3-s}{2},\frac{\alpha}{2}+2,\alpha+2;1
\right).
\label{he4single}
\end{multline}
\end{theorem}
\begin{proof}
See Section \ref{se:proofsbet124}.
\end{proof}
The ${}_3F_2$ evaluated at $z=1$ is known as \textit{Clausen's series} and has been studied in the special functions literature \cite{C21,MP12}. These formulae simplify further if we consider the negative integer moments $s=k$ with $k \in \mathbb{N}$, since the hypergeometric series in \eqref{hesymp}, \eqref{heortho} terminate and consists of only $k-1$ terms.
\begin{corollary}[Integer case for $\beta=1$ and $\beta=4$]
	\label{cor:int}
Consider the limiting moments defined in \eqref{limit-exists} for any $k \in \mathbb{N}$ with $k<\alpha+1$. Then we have
\begin{equation}
\begin{split}
\mathcal{M}^{(4)}(-k,\alpha) = \,& 8^{k-1}\frac{\Gamma(\alpha-k+1)\Gamma(k-\frac{1}{2})}{\sqrt{\pi}k!\Gamma(\alpha+k)}\\
&-\frac{2^{k-1}}{(k-2)!}\sum_{j=0}^{k-2}\binom{k-2}{j}\frac{(-1)^{j}\Gamma(\alpha-k+1+j)}{(\alpha+2+2j)\Gamma(\alpha+2+j)},
\end{split}
\end{equation}
and
\begin{equation}
\begin{split}
\mathcal{M}^{(1)}(-k,\alpha) = \,& 8^{k-1}\frac{2\Gamma(2\alpha-k+1)\Gamma(k-\frac{1}{2})}{\sqrt{\pi}k!\Gamma(2\alpha+k)}\\
&-\frac{2^{k-1}}{(k-2)!}\sum_{j=0}^{k-2}\binom{k-2}{j}\frac{(-1)^{k+j}\Gamma(2\alpha-1-j)}{(1+j-\alpha)\Gamma(2\alpha+k-j)},
\end{split}
\end{equation}
where the sums are set equal to zero if $k=1$.
\end{corollary}

\begin{remark}
We also establish linear first order recurrence relations for $\mathcal{M}^{(4)}(-k,\alpha)$ and $\mathcal{M}^{(1)}(-k,\alpha)$, see Section \ref{se:rec}, equation \eqref{rec}.
\end{remark}

\subsection{Main results for general {\texorpdfstring{$\beta > 0$}{beta}}}
\label{se:betaresults}
For general $\beta>0$, the approaches used for the exactly solvable cases $\beta \in \{1,2,4\}$ mentioned before, particularly the orthogonal polynomial technique, cannot be applied. Instead, formulae for \eqref{momsintro} have been obtained as sums over partitions \cite{FD16,MRW17} or as nested contour integrals \cite[Appendix I]{FD16}. In \cite{FD16}, such moments were studied in the context of the replica method which involves the formal limit $N \to 0$ of the moments. 

We state the next result as a corollary as it follows from \cite{FD16} without too much work. Let $\eta$ be a partition of an integer $k$ having length $\ell(\eta)$ and let $\eta_{1},\eta_{2},\ldots$ denote the individual parts. The weight of $\eta$ is defined as $|\eta| := \sum_{i=1}^{\ell(\eta)}\eta_{i}$. We adopt the convention that the parts are weakly decreasing $\eta_{1} \geq \eta_{2} \geq \ldots \geq \eta_{\ell(\eta)}\geq 1$. 
\begin{corollary}
\label{thm:momsbeta}
The limit \eqref{limit-exists} exists for any $\beta>0$ and $s=k \in\mathbb{N}$ with $k<\alpha+1$. We have
\begin{equation}
\mathcal{M}^{(\beta)}(-k,\alpha) = \sum_{\eta, |\eta|=k}A^{(\beta)}_{\eta}\alpha^{-}_{\eta} \label{betform}
\end{equation}
where the sum is taken over all partitions $\eta$ of the integer $k$. The coefficient $A^{(\beta)}_{\eta}$ is given in terms of $\kappa = \frac{\beta}{2}$ as
\begin{equation}
\begin{split}
A^{(\beta)}_{\eta} &:= \frac{k(\eta_{1}-1)!\kappa^{k}} {(\kappa(\ell(\eta)-1)+1)_{\eta_1}} \prod_{i=2}^{\ell(\eta)} \frac{(\kappa(1-i))_{\eta_{i}}} {(\kappa(\ell(\eta)-i)+1)_{\eta_{i}}}
\prod_{i=1}^{\ell(\eta)}\frac{1}{(\kappa(\ell(\eta)-i+1))_{\eta_i}}\\
&\times \prod_{1 \leq i < j \leq \ell(\eta)}\left(\frac{\kappa(j-i)+\eta_{i}-\eta_{j}}{\kappa(j-i)}\right)\frac{(\kappa(j-i+1))_{\eta_{i}-\eta_{j}}}{(\kappa(j-i-1)+1)_{\eta_{i}-\eta_{j}}} \label{AlargeN}
\end{split}
\end{equation}
and
\begin{equation} \label{e:a2}
\alpha^{-}_{\eta} := \prod_{i=1}^{\ell(\eta)}(\alpha+1+\kappa(i-1))_{-\eta_i}.
\end{equation}
In this formula the rising factorial is $(x)_{n} = x(x+1)\ldots(x+n-1)$ and $(x)_{-n} = \frac{1}{(x-1)(x-2)\ldots(x-n)}$.
\end{corollary}
\begin{proof}
    See Section \ref{sec:fdresults}.
\end{proof}
Here we present some lower order moments as examples, obtained from \eqref{betform} with the help of computer algebra.
\begin{equation}\label{e:limitMoments-intro}
\begin{split}
\mathcal{M}^{(\beta)}(-1,\alpha) &= \frac{1}{\alpha}\\
\mathcal{M}^{(\beta)}(-2,\alpha) &= \frac{\beta}{\alpha(\alpha-1)(2\alpha+\beta)}\\
\mathcal{M}^{(\beta)}(-3,\alpha) &= \frac{\beta^{2}}{\alpha(\alpha-1)(\alpha-2)(2\alpha+\beta)(\alpha+\beta)}
\end{split}
\end{equation}
and
\begin{equation}
\begin{split}
&\mathcal{M}^{(\beta)}(-4,\alpha) \\
&= \frac{\beta^{3}(3\beta+5\alpha-6)}{\alpha(\alpha-1)(\alpha-2)(\alpha-3)(2\alpha+\beta)(\alpha+\beta)(2\alpha+3\beta)(2\alpha-2+\beta)}.
\label{e:limitMoments4}
\end{split}
\end{equation}
\begin{remark}
These moments satisfy a duality under the mapping $\beta \to \frac{4}{\beta}$ in the sense that the following identity holds
\begin{equation}
\mathcal{M}^{(\beta)}(-k,\alpha) = \left(-\frac{2}{\beta}\right)\mathcal{M}^{(\frac{4}{\beta})}\left(-k,-\frac{2\alpha}{\beta}\right).
\end{equation}
This duality is a consequence of \cite[equation (90)]{FD16}, see also \cite{F22}.
\end{remark}
Our final result shows a relation between these moments and the Bessel zeta function. Of particular interest here are the even integer evaluations of $\zeta_{\nu}(s)$. As for $\zeta(s)$, \eqref{besselzeta} has known explicit values when $s=2k$ is an even integer,
\begin{equation}
\begin{split}
\zeta_{\nu}(2) &= \frac{1}{2^{2}(\nu+1)}\\
\zeta_{\nu}(4) &= \frac{1}{2^{4}(\nu+1)^{2}(\nu+2)}\\
\zeta_{\nu}(6) &= \frac{1}{2^{5}(\nu+1)^{3}(\nu+2)(\nu+3)}\\
\zeta_{\nu}(8) &= \frac{5\nu+11}{2^{8}(\nu+1)^{4}(\nu+2)^{2}(\nu+3)(\nu+4)}\\
&\vdots
\label{besselzetmoments}
\end{split}
\end{equation}
etc. These formulas are attributed to Rayleigh \cite[Chapter 15]{W44}. There are known recursion relations which allow the explicit computation of these quantities to any desired order. To see the connection with the present study, first note the structural similarity with formulae \eqref{e:limitMoments-intro} and \eqref{e:limitMoments4}. Second, substituting the scaling $\alpha = \frac{\beta}{2}(\nu+1)$ into \eqref{e:limitMoments-intro} and \eqref{e:limitMoments4} and extracting the leading term as $\beta \to \infty$ exactly reproduces the formulae in \eqref{besselzetmoments}. We show in the following that this holds for all higher moments.

\begin{theorem}
\label{thm:lowtemp}
Let $\alpha = \frac{\beta}{2}(\nu+1)$ in \eqref{laguerrebeta} and consider the inverse moments \eqref{momsintro}. For any sequence $\beta:= \beta_{N}$ such that $\beta_{N} \to \infty$ as $N \to \infty$ and any $k \in \mathbb{N}$, we have
\begin{equation}
\lim_{N \to \infty}\left(\frac{\beta_{N}}{8N}\right)^{k}\mathcal{M}_{N}^{(\beta_{N})}\left(-k,\frac{\beta_{N}}{2}(\nu+1)\right) = \zeta_{\nu}(2k), \label{largebetmoms}
\end{equation}
where $\zeta_{\nu}(2k)$ is the Bessel zeta function \eqref{besselzeta}.
\end{theorem}
\begin{proof}
    See Section \ref{sec:fdresults}.
\end{proof}
\section*{Acknowledgements} 
A. M. and N. S. are grateful for support from the Royal Society, grant URF$\backslash$R$\backslash$221017 and URF\textbackslash R\textbackslash231028 respectively.

\section{Overview of approach and proofs for {\texorpdfstring{$\beta \in \{1,2,4\}$}{beta = 1,2,4}}}
\label{sec:proofs}
We begin by introducing the one-point eigenvalue distribution, defined as
\begin{equation}
\rho^{(\beta)}_{N}(x) = N\int_{0}^{\infty}\ldots \int_{0}^{\infty}P_{\beta,\alpha}(x,\lambda_{2},\ldots,\lambda_N)d\lambda_{2}\ldots d\lambda_{N} \label{one-point}
\end{equation}
where $P_{\beta,\alpha}(\lambda_1,\ldots,\lambda_{N})$ is the joint probability density function of eigenvalues given by \eqref{laguerrebeta}. From \eqref{one-point} and the permutation invariance of \eqref{laguerrebeta}, we have
\begin{equation}
\mathcal{M}_{N}^{(\beta)}(-k,\alpha) = \int_{0}^{\infty}x^{-k}\rho^{(\beta)}_{N}(x)dx. \label{momsrho}
\end{equation}
For $\beta \in \{1,2,4\}$, the integral in \eqref{one-point} has known expressions in terms of Laguerre polynomials. The simplest case is $\beta=2$ where we have
\begin{equation}
\rho^{(2)}_{N}(x) = x^{\alpha}e^{-x}\sum_{j=0}^{N-1}\frac{L^{(\alpha)}_{j}(x)^{2}}{h_{j}}, \qquad h_{j} = \frac{\Gamma(j+\alpha+1)}{j!}, \label{bet2exact}
\end{equation}
where $L^{(\alpha)}_{j}(x)$ is the Laguerre polynomial of degree $j$. By the Christoffel-Darboux formula for orthogonal polynomials, this can be further simplified as 
\begin{equation}
\rho^{(2)}_{N}(x) = x^{\alpha}e^{-x}c_{N,\alpha}\left(L^{(\alpha)}_{N-1}(x)'L^{(\alpha)}_{N}(x)-L^{(\alpha)}_{N}(x)'L^{(\alpha)}_{N-1}(x)\right), \label{rho2laguerre}
\end{equation}
where the prime denotes differentiation with respect to $x$, and $c_{N,\alpha} = \frac{N!}{\Gamma(N+\alpha)}$. See \cite[Chapter 5]{Forr} for background on formulae \eqref{bet2exact} and \eqref{rho2laguerre}. The cases $\beta=1$ and $\beta=4$ are slightly more complicated, requiring a modification of the standard orthogonal polynomial approach in terms of skew-orthogonal polynomials. For $\beta=4$, we have
\begin{equation}
	\rho^{(4)}_{N}(x) = \frac{1}{2}\rho^{(2)}_{2N}(x) - \frac{(2N)!}{4\Gamma(2N+\alpha)}x^{\frac{\alpha}{2}-1}e^{-\frac{x}{2}}L^{(\alpha-1)}_{2N}(x)
    \int_{0}^{x}t^{\frac{\alpha}{2}}e^{-\frac{t}{2}}L^{(\alpha+1)}_{2N-1}(t)dt
	\label{bet4exact}
\end{equation}
and for $\beta=1$ and $N$ even, we have
\begin{equation}
	\label{bet1exact}
	\begin{split}
		\rho^{(1)}_{N}(x) &= 2\rho^{(2)}_{N}(2x)|_{\alpha \to 2\alpha}-\frac{N!}{2\Gamma(N+2\alpha)}(2x)^{\alpha}e^{-x}L^{(2\alpha+1)}_{N-1}(2x)\\
		&\times 2^{\alpha}\int_{0}^{\infty}\mathrm{sgn}(x-t)t^{\alpha-1}e^{-t}L^{(2\alpha-1)}_{N}(2t)dt,
	\end{split}
\end{equation}
where $\mathrm{sgn}(x)$ is the sign of $x$. In the case $\beta=1$ we will assume $N$ is even throughout for simplicity. Also note that \eqref{bet4exact} and \eqref{bet1exact} seem ill-defined when $-1 < \alpha \leq 0$, but this can be resolved by a suitable analytic continuation, as discussed in \cite{W99}. See \cite{W99,AFNvM00} and \cite[Chapter 6]{Forr} for a derivation of \eqref{bet4exact} and \eqref{bet1exact}.

The hard edge scaling regime corresponds to computing the limiting profiles of these densities in a neighbourhood of the origin of order $1/N$. More precisely, applying suitable asymptotics for Laguerre polynomials (see Lemma \ref{lem:lag} for the required asymptotics), the following limit is well known
\begin{equation}
\lim_{N \to \infty}\frac{1}{4N}\rho_{N}^{(\beta)}(u/4N)	= \rho^{(\beta)}_{\mathrm{HE}}(u) \label{limHE}
\end{equation}
uniformly in compact subsets of $u \in [0,\infty)$. The limiting formulae are given in terms of Bessel functions. For $\beta=2$, we have
\begin{equation}
\rho^{(2)}_{\mathrm{HE}}(u)	= \frac{1}{4}\left(J_{\alpha}(\sqrt{u})^{2}-J_{\alpha+1}(\sqrt{u})J_{\alpha-1}(\sqrt{u})\right),
\label{besselbeta2}
\end{equation}
and for $\beta=4$ and $\beta=1$, we have
\begin{align}
\rho^{(4)}_{\mathrm{HE}}(u) &= \rho^{(2)}_{\mathrm{HE}}(2u)-\frac{1}{4\sqrt{2u}}J_{\alpha-1}(\sqrt{2u})\int_{0}^{\sqrt{2u}}J_{\alpha+1}(t)dt, \label{bet4he}\\
	\rho^{(1)}_{\mathrm{HE}}(u) &= 2\rho^{(2)}_{\mathrm{HE}}(2u)|_{\alpha \to 2\alpha}+\frac{J_{2\alpha+1}(\sqrt{2u})}{2\sqrt{2u}}\int_{\sqrt{2u}}^{\infty}J_{2\alpha-1}(t)dt. \label{bet1he}
\end{align}
See \cite[Chapter 7]{Forr} for background on \eqref{limHE}-\eqref{bet1he}.
\subsection{Proof of Theorems \ref{thm:bet2} and \ref{thmbet1bet4}}
\label{se:proofsbet124}
We begin with the following tail estimate showing that the dominant contribution to the moments \eqref{momsrho} comes from the hard edge scaling regime mentioned above.
\begin{lemma}
	\label{lem:cutoff}
	Fix $\alpha > -\frac{1}{2}$. For any $\beta \in \{1,2,4\}$ and $s \in \mathbb{C}$ belonging to the strip $\frac{1}{2} < \mathrm{Re}(s) < a+1$, for all $M$ sufficiently large, we have
	\begin{equation}
		\bigg{|}\int_{M/N}^{\infty}x^{-s}\rho^{(\beta)}_{N}(x)dx\bigg{|} \leq c|N^{s}M^{-s+\frac{1}{2}}| \label{orderbound}
	\end{equation}
	where $c>0$ is a constant independent of $M$ and $N$. Consequently
	\begin{equation}
		\lim_{M \to \infty}\lim_{N \to \infty}N^{-s}\int_{M/N}^{\infty}x^{-s}\rho^{(\beta)}_{N}(x)dx = 0.
	\end{equation}
\end{lemma}
We postpone the proof of this lemma to Section \ref{sec:tail}. The above illustrates why the Marchenko-Pastur limit \eqref{mplaw} does not give the correct order of the moments, since \eqref{mplaw} would predict a leading term of order $N^{1-s}$ which is sub-leading to $N^{s}$ on the right-hand side of \eqref{orderbound} for every $\frac{1}{2} < \mathrm{Re}(s) < \alpha+1$. Lemma \ref{lem:cutoff} implies
\begin{equation}
	\mathcal{M}^{(\beta)}_{N}(-s,\alpha)  \sim \int_{0}^{M/N}x^{-s}\rho^{(\beta)}_{N}(x)dx
\end{equation}
where $\sim$ indicates that the ratio of the left and right hand sides above tends to $1$ as $N \to \infty$ followed by $M \to \infty$. Changing variable $x=\frac{u}{4N}$, we write
\begin{equation}
	\mathcal{M}^{(\beta)}_{N}(-s,a) \sim N^{s}4^{s}\int_{0}^{4M}u^{-s}\frac{1}{4N}\rho^{(\beta)}_{N}(u/4N)du. \label{momsdef}
\end{equation}
Since $\frac{1}{4N}\rho^{(\beta)}_{N}(u/4N)$ converges uniformly on compact subsets of $u \in [0,\infty)$ to the hard edge density $\rho^{(\beta)}_{\mathrm{HE}}(u)$, taking the subsequent limit $M \to \infty$, we get
\begin{equation}
	\lim_{N \to \infty}N^{-s}\mathcal{M}^{(\beta)}_{N}(-s,\alpha) = 4^{s}\int_{0}^{\infty}u^{-s}\rho^{(\beta)}_{\mathrm{HE}}(u)du. \label{bet2conv}
\end{equation}
\begin{proof}[Proof of Theorem \ref{thm:bet2}]

	Substituting \eqref{besselbeta2} into the right-hand side of \eqref{bet2conv}, we change variable $u \to u^{2}$ and make use of the identity (see e.g.\,\cite{M01}),
	\begin{equation}
		\begin{split}
			\label{mellinbessels}
			E_{a,\gamma}(s)&:=\int_{0}^{\infty}u^{s-1}J_{a}(u)J_{\gamma}(u)du \\
			&= \frac{\Gamma(1-s)\Gamma(\frac{a+\gamma+s}{2})}{2^{1-s}\Gamma(\frac{\gamma-a-s+2}{2})\Gamma(\frac{a+\gamma-s+2}{2})\Gamma(\frac{a-\gamma-s+2}{2})}.
		\end{split}
	\end{equation}
This gives
\begin{equation}
\begin{split}
\int_{0}^{\infty}u^{-s}\rho^{(2)}_{\mathrm{HE}}(u)du &= \frac{1}{2}\left(E_{\alpha,\alpha}(2(1-s))-E_{\alpha-1,\alpha+1}(2(1-s))\right)\\
&= \frac{4^{1-s}\Gamma(\alpha+1-s)\Gamma(2(s-1))}{2s(s-1)\Gamma(s+\alpha)\Gamma(s-1)^{2}}.
\end{split}
\end{equation}
	Applying the duplication formula for the Gamma function and substituting the result into \eqref{bet2conv} completes the proof.
\end{proof}
We now show how to derive the same result using Theorem \ref{thm:cmos}.
\begin{proof}[Second proof of Theorem \ref{thm:bet2}]
By the reflection property given in Theorem \ref{thm:cmos}, we have
\begin{equation}
	\mathcal{M}^{(2)}_{N}(-s,\alpha) = \mathcal{M}^{(2)}_{N}(s-1,\alpha)\frac{\Gamma(\alpha+1-s)}{\Gamma(s+\alpha)}. \label{reflectidentity}
\end{equation}
Then we apply \eqref{mplaw} with $f(x) = x^{s-1}$ for $\mathrm{Re}(s) > \frac{1}{2}$, which holds since the convergence to the Marchenko-Pastur law is also known to hold for positive moments, see e.g.\ \cite{DE06}. This implies 
\begin{equation}
	\begin{split}
		\label{mpcon}
\mathcal{M}^{(2)}_{N}(s-1,\alpha) &\sim N^{s}\frac{1}{2\pi}\int_{0}^{4}\frac{\sqrt{4-x}}{\sqrt{x}}x^{s-1}dx\\
&=N^{s}\frac{1}{\sqrt{\pi}}4^{s-1}\frac{\Gamma(s-\frac{1}{2})}{\Gamma(s+1)}.
\end{split}
\end{equation}
Inserting \eqref{mpcon} into \eqref{reflectidentity} completes the proof.
\end{proof}

\begin{proof}[Proof of Theorem \ref{thmbet1bet4}]
\label{lem:cor1}
We start with $\beta=4$. Substituting \eqref{bet4he} into \eqref{bet2conv}, the first term is explicit as we have already dealt with the $\beta=2$ case above. It remains to deal with the correction term in \eqref{bet4he}, which after a suitable change of variables gives rise to the quantity
\begin{equation}
	I_{\alpha}(s) := 2^{s+\frac{1}{2}}\int_{0}^{\infty}u^{-2s}J_{\alpha-1}(u)\int_{0}^{u}J_{\alpha+1}(t)dt\,du. \label{ialph}
\end{equation}
We recall that the Bessel function can be written as
\begin{equation}
J_{\alpha}(u) = \frac{\left(\frac{u}{2}\right)^{\alpha}}{\Gamma(\alpha+1)}{}_0F_1\left(\alpha+1;-\frac{u^{2}}{4}\right). \label{besselhyper}
\end{equation}
For general $\alpha$ and non-negative $u$, we have
\begin{equation}
\int_{0}^{u}J_{\alpha}(t)dt = \frac{u^{\alpha+1}}{2^{\alpha}}\frac{1}{\Gamma(\alpha+2)}{}_1F_2\left(\frac{\alpha+1}{2};\alpha+1,\frac{\alpha+3}{2};-\frac{u^{2}}{4}\right).  \label{intbesselhyper}
\end{equation}
This follows from writing the series expansion of $J_{\alpha}(x)$ and integrating term by term, then recognising the resulting series expansion as the advertised ${}_1F_2$. Now inserting \eqref{besselhyper} and \eqref{intbesselhyper} into \eqref{ialph}, we recognise the Mellin transform of a product of two hypergeometric functions. This particular transform has been evaluated in \cite[Eq. (2.1a)]{M01} and using this we obtain
\begin{equation}
I_{\alpha}(s) = c_{\alpha}(s)\,{}_3F_2\left(\frac{\alpha+2}{2},\alpha+1-s,2-s;\frac{\alpha+4}{2},\alpha+2; 1\right),
\end{equation}
where 
\begin{equation}
    c_{\alpha}(s) = 2^{-s+\frac{3}{2}}\frac{\Gamma(\alpha+1-s)}{\Gamma(s-1)\Gamma(\alpha+3)}.
\end{equation}
This completes the proof of \eqref{hesymp}.

Now for $\beta=1$, we consider the hard edge profile in \eqref{bet1he}, which we write as
\begin{equation}
 \rho^{(1)}_{\mathrm{HE}}(x) = 2\rho^{(2)}_{\mathrm{HE}}(2x)|_{\alpha \to 2\alpha}+\frac{1}{2}\frac{J_{2\alpha+1}(\sqrt{2x})}{\sqrt{2x}}-\frac{J_{2\alpha+1}(\sqrt{2x})}{2\sqrt{2x}}\int_{0}^{\sqrt{2x}}J_{2\alpha-1}(t)dt. \label{hebet1rewrite}
\end{equation}
where we employed the definite integral $\int_{0}^{\infty}J_{2\alpha-1}(t)dt = 1$, see e.g.\ \cite[10.22.41]{NIST:DLMF}. The required integral transform of the first and third terms of \eqref{hebet1rewrite} follow identically to the previous computations and yield the first and third terms in \eqref{heortho}, while for the middle term, we have
\begin{equation}
\int_{0}^{\infty}u^{-s}\frac{J_{2\alpha+1}(\sqrt{2u})}{2\sqrt{2u}}du = 2^{-s-1}\frac{\Gamma(\alpha+1-s)}{\Gamma(\alpha+1+s)},
\end{equation}
see e.g.\ \cite[10.22.43]{NIST:DLMF}. This completes the proof of \eqref{heortho}.

Finally, we prove the alternative representation \eqref{he4single}. The computations below are valid for $\mathrm{Re}(s)>\frac{1}{2}$, where the Gauss summation formula applies and the hypergeometric series converge absolutely at $1$; the final identities then extend to all $s \in \mathcal{S}$ by analytic continuation. Starting from \eqref{hesymp}, we use the formula for $\mathcal{M}^{(2)}(-s,\alpha)$ from Theorem~\ref{thm:bet2} together with the duplication formula and Gauss summation (see e.g.\ \cite[15.4.20]{NIST:DLMF}) to arrive at 
\[
2^{s-1}\mathcal{M}^{(2)}(-s,\alpha)
=
\frac{2^{s-1}\Gamma(\alpha+1-s)}{\Gamma(\alpha+3)\Gamma(s-1)}
\frac{\alpha+2}{s-1}
\,{}_2F_1(\alpha+1-s,2-s;\alpha+2;1).
\]
Substituting this into \eqref{hesymp} yields
\begin{multline}
\mathcal{M}^{(4)}(-s,\alpha) =
\frac{2^{s-1}\Gamma(\alpha+1-s)}{\Gamma(\alpha+3)\Gamma(s-1)}
\left[
\frac{\alpha+2}{s-1}{}_2F_1(\alpha+1-s,2-s;\alpha+2;1)
\right.\\
\left.-{}_3F_2\!\left(\frac{\alpha}{2}+1,\alpha+1-s,2-s;\frac{\alpha}{2}+2,\alpha+2;1\right)
\right].
\end{multline}
Expanding both series gives
\begin{equation}\label{e:0}
\mathcal{M}^{(4)}(-s,\alpha)
=
\frac{2^{s-1}\Gamma(\alpha+1-s)}{\Gamma(\alpha+3)\Gamma(s-1)}
\sum_{n=0}^{\infty}
\frac{(\alpha+1-s)_n(2-s)_n}{(\alpha+2)_n\,n!}
\left(
\frac{\alpha+2}{s-1}
-
\frac{\frac{\alpha}{2}+1}{\frac{\alpha}{2}+1+n}
\right).
\end{equation}
A direct calculation shows that
\begin{equation}\label{e:1}
\frac{\alpha+2}{s-1}
-
\frac{\frac{\alpha}{2}+1}{\frac{\alpha}{2}+1+n}
=
\frac{\alpha+3-s}{s-1}
\frac{\left(\frac{\alpha+5-s}{2}\right)_n\left(\frac{\alpha}{2}+1\right)_n}
{\left(\frac{\alpha+3-s}{2}\right)_n\left(\frac{\alpha}{2}+2\right)_n}.
\end{equation}
Substituting \eqref{e:1} back into \eqref{e:0} gives
the desired result.
\end{proof}

\subsection{Recurrence relations}
\label{se:rec}
In this section, we show that the limiting moments $\mathcal{M}^{(1)}(-k,\alpha)$ and $\mathcal{M}^{(4)}(-k,\alpha)$ satisfy a first order recurrence relation. By solving this recurrence relation, we deduce exact formulae for the moments as alternative expressions to those given in Corollary \ref{cor:int}.
\begin{proposition}
	For any integer $k \geq 1$, we have the following explicit formulae in the cases $\beta=1$ and $\beta=4$: 
	\begin{equation}
		\begin{split}
			&\mathcal{M}^{(1)}(-k,\alpha)  = 2^{k-1}\frac{\Gamma(\alpha+1-k)}{\Gamma(\alpha+1+k)}\\
			&\times \left(\alpha+1-3\sum_{m=1}^{k-1}4^{m-1}\frac{\Gamma(m-\frac{1}{2})\Gamma(2\alpha-m+2)\Gamma(\alpha+1+m)}{\sqrt{\pi}\Gamma(m+2)\Gamma(2\alpha+m+1)\Gamma(\alpha+1-m)}\right) \label{beta1solrec}
		\end{split}
	\end{equation}
	and
	\begin{equation}
		\begin{split}
			&\mathcal{M}^{(4)}(-k,\alpha)  = 2^{k-2}\frac{\Gamma(\frac{\alpha}{2}-k)}{\Gamma(\frac{\alpha}{2}+k)}\\
			&\times \left(\frac{\alpha}{2}-1-3\sum_{m=1}^{k-1}4^{m-1}\frac{\Gamma(m-\frac{1}{2})\Gamma(\frac{\alpha}{2}+m)\Gamma(\alpha-m)}{\sqrt{\pi}\Gamma(m+2)\Gamma(\frac{\alpha}{2}-m)\Gamma(\alpha-1+m)}\right). \label{beta4solrec}
		\end{split}
	\end{equation}
	The sums are interpreted as $0$ when $k=1$.
\end{proposition}
\begin{proof}
	We begin with $\beta=1$. For ease of notation, set $m_{k} =  \mathcal{M}^{(1)}_{N}(-k,\alpha)$. By \cite[Theorem 3.5]{CMSV16}, we have the recurrence relation
	\begin{equation}
		\begin{split}
			&m_{k-1}- 2(\alpha+N)m_{k}+(\alpha^{2}+\alpha-k(k+1))m_{k+1}\\
			&= -\frac{3(2^{k+1})}{4(k+1)}\left((2\alpha+2N+k)v_{k}-v_{k-1}\right),
		\end{split}
	\end{equation}
	where $v_{k} = \mathcal{M}^{(2)}_{N-1}(-k,2\alpha+1)$. Since $\mathcal{M}^{(\beta)}_{N}(-k,\alpha)$ is of order $N^{k}$ as $N \to \infty$, we can extract the leading order terms above and obtain the limiting recurrence
	\begin{equation}
		\begin{split}
			&-2\mathcal{M}^{(1)}(-k,\alpha)+(\alpha^{2}+\alpha-k(k+1))\mathcal{M}^{(1)}(-(k+1),\alpha)\\
			&=-\frac{6(2^{k+1})}{4(k+1)}\mathcal{M}^{(2)}(-k,2\alpha+1).
		\end{split}
	\end{equation}
	This is in the form of a first order recurrence relation with an inhomogeneous right-hand term. We set $a_{k} = \mathcal{M}^{(1)}(-k,\alpha)$ and write the recurrence as
	\begin{equation}
		a_{k+1}=a_{k}f_{k}+g_{k} \label{rec}
	\end{equation}
	where $f_{k} = \frac{2}{\alpha^{2}+\alpha-k(k+1)}$ and 
	\begin{equation}
		g_{k} = -\frac{6(2^{k+1})}{4(k+1)(\alpha^{2}+\alpha-k(k+1))}\mathcal{M}^{(2)}(-k,2\alpha+1).
	\end{equation}
	The solution of this recurrence with initial condition $a_{1}=\frac{1}{\alpha}$ is
	\begin{equation}
		\begin{split}
			a_{k} &= \left(\prod_{i=1}^{k-1}f_{i}\right)\left(\frac{a_{1}}{f_{0}}+\sum_{m=1}^{k-1}\frac{g_{m}}{\prod_{i=0}^{m}f_{i}}\right)\\
			&= 2^{k-1}\prod_{i=0}^{k-1}\frac{1}{(\alpha^{2}+\alpha-i(i+1))}\\
			&\times\left(\alpha+1-\sum_{m=0}^{k-1}\frac{3\mathcal{M}^{(2)}(-m,2\alpha+1)}{(m+1)(\alpha^{2}+\alpha-m(m+1))}\prod_{i=0}^{m}(\alpha^{2}+\alpha-i(i+1))\right).
		\end{split}
	\end{equation}
	To simplify this, note that
	\begin{equation}
		\prod_{i=0}^{k-1}(\alpha^{2}+\alpha-i(i+1)) = \prod_{i=0}^{k-1}(\alpha-i)(\alpha+i+1) = \frac{\Gamma(\alpha+1+k)}{\Gamma(\alpha+1-k)}.
	\end{equation}
	So
	\begin{equation}
		\begin{split}
			&\mathcal{M}^{(1)}(-k,\alpha) = 2^{k-1}\frac{\Gamma(\alpha+1-k)}{\Gamma(\alpha+1+k)}\\
			&\times \left(\alpha+1-\sum_{m=0}^{k-1}\frac{3\mathcal{M}^{(2)}(-m,\alpha)}{(m+1)(\alpha-m)(\alpha+m+1)}\frac{\Gamma(\alpha+2+m)}{\Gamma(\alpha-m)}\right). \label{solak}
		\end{split}
	\end{equation}
	Recalling the exact result of Theorem \ref{thm:bet2} and inserting this into \eqref{solak} completes the proof of \eqref{beta1solrec}. Then \eqref{beta4solrec} follows from the duality $\beta \to \frac{4}{\beta}$ and $\alpha \to -\frac{2\alpha}{\beta}$ and applying the reflection identity for the Gamma functions.
\end{proof}

\subsection{Tail estimates}
\label{sec:tail}
The goal of this section is to prove Lemma \ref{lem:cutoff}. We start with the required asymptotics of Laguerre polynomials.
\begin{lemma}
	\label{lem:lag}
	Let $M>0$ be a constant, chosen sufficiently large and $\delta>0$ be chosen sufficiently small. On the interval $\frac{M}{N} \leq x \leq \delta N$, we have for any fixed $\alpha>-2$,
	\begin{equation}
		x^{\frac{\alpha}{2}}e^{-\frac{x}{2}}L_{N}^{(\alpha)}(x) = O(N^{\frac{\alpha}{2}}(Nx)^{-\frac{1}{4}}), \qquad N \to \infty, \label{lag-main-bnd}
	\end{equation}
	where the big-O term is uniform in the sense that the implied constant is independent of $M$, $N$ and $x$. Furthermore, uniformly on $0 \leq x \leq M$, we have
	\begin{equation}
		\lim_{N \to \infty}N^{-\alpha}x^{\frac{\alpha}{2}}L^{(\alpha)}_{N}\left(\frac{x}{N}\right) = J_{\alpha}(2\sqrt{x}). \label{hilb}
	\end{equation}
\end{lemma}
\begin{proof}
	By \cite[18.15.19]{NIST:DLMF} Laguerre polynomials with fixed $\alpha>-1$ satisfy the following asymptotic expansion as $N \to \infty$,
	\begin{equation}
		(\nu x)^{\frac{\alpha}{2}}e^{-\frac{1}{2}\nu x}L_N^{(\alpha)}(\nu x) = \frac{\nu^{\frac{\alpha}{2}}}{2^\alpha x^{\frac{1}{4}}(1-x)^{\frac{1}{4}}} \left[ \xi^{1/2} J_\alpha(\nu \xi) + \epsilon_1 \right], \label{lagasy}
	\end{equation}
	where $\nu = 4N + 2\alpha + 2$, $\xi = \frac{1}{2} \left( \arcsin\sqrt{x} + \sqrt{x(1-x)} \right)$ and the error term $\epsilon_1 = O(\xi^{1/2} \mathrm{env}J_\alpha(\nu \xi)/\nu)$ uniformly on $0\leq x < 1-\delta$, where
	\begin{equation} \mathrm{env}J_\alpha(x) = \begin{cases} J_{\alpha}(x), & x < j_{\alpha,1} \\ \sqrt{J_{\alpha}(x)^{2}+Y_{\alpha}(x)^{2}}, & x \geq j_{\alpha,1} \end{cases} \end{equation}
	and $Y_{\alpha}(x)$ is the corresponding Hankel function, $j_{\alpha,1}$ is the first zero of $J_{\alpha}(x)$.
	Since $0 \leq x \leq \delta$, we can choose $\delta$ sufficiently small such that $\frac{1}{2}\sqrt{x} \leq \xi \leq 2\sqrt{x}$, which implies $\nu\xi \geq 2\sqrt{M}$. Choosing $M$ sufficiently large and using the standard large argument asymptotics of Bessel functions (see \cite[10.17.3]{NIST:DLMF}), we have $|J_{\alpha}(\nu \xi)| = O((\nu \xi)^{-\frac{1}{2}})$. Hence the main term in \eqref{lagasy} including pre-factors is $O(\nu^{\frac{\alpha}{2}}(\nu \sqrt{x})^{-\frac{1}{2}})$. Similarly, since $\nu \xi$ is large, $\mathrm{env}J_\alpha(\nu \xi) = O((\nu \xi)^{-\frac{1}{2}})$. We conclude that the error term in \eqref{lagasy} including pre-factors is $O(\nu^{\frac{\alpha}{2}-1}(\nu \sqrt{x})^{-\frac{1}{2}})$. Replacing $x$ with $\frac{x}{\nu}$ in these bounds completes the proof. Finally, we address the case $L^{(\alpha)}_{N}(x)$ with $\alpha \in (-2,-1]$. Using the identity (e.g.\ \cite[18.9.13]{NIST:DLMF}) $L^{(\alpha)}_{N}(x) = L^{(\alpha+1)}_{N}-L^{(\alpha+1)}_{N-1}(x)$, we apply expansion \eqref{lagasy} with $x$ replaced with $x/\nu$ and note that the the leading terms cancel yielding the same expansion \eqref{lag-main-bnd}. The asymptotics \eqref{hilb} are well known and follow from the approximation \eqref{lagasy}, we omit the details.
\end{proof}

\begin{proof}[Proof of Lemma \ref{lem:cutoff}]
	We first consider the contribution of the integral for $x > \delta N$ for some small $\delta>0$. Changing variable $x=uN$ we get,
	\begin{equation}
		\int_{\delta N}^{\infty}x^{-s}\rho^{(\beta)}_{N}(x)dx = N^{1-s}\int_{\delta}^{\infty}u^{-s}\rho^{(\beta)}_{N}(uN)du = O(N^{1-s}),
	\end{equation}
	where the above estimate follows from the weak convergence \eqref{mplaw}. Here, and below, the big-O term with complex exponent is interpreted by taking the absolute value, which is equivalent to taking the real part of the exponent. Hence, the above is negligible since $N^{1-s}$ is of lower order than $N^{s}$ for $\mathrm{Re}(s)>\frac{1}{2}$.

	Now we estimate the range $M/N < x < \delta N$. We do this case by case, starting with $\beta=2$. Recall the representation of $\rho^{(2)}_{N}(x)$ in \eqref{rho2laguerre}. Using the identity (see e.g.\ \cite[18.9.23]{NIST:DLMF}) $L^{(\alpha)}_{N}(x)' = -L_{N-1}^{(\alpha+1)}(x)$, we have
	\begin{equation}
		\rho^{(2)}_{N}(x) = x^{\alpha}e^{-x}c_{N,\alpha}\left(L^{(\alpha+1)}_{N-1}(x)L^{(\alpha)}_{N-1}(x)-L^{(\alpha)}_{N}(x)L_{N-2}^{(\alpha+1)}(x)\right). \label{bad-dens}
	\end{equation}
	We note that this identity is not yet adapted to finding asymptotics as $N \to \infty$, since the leading terms obtained from Lemma \ref{lem:lag} cancel out when substituted into \eqref{bad-dens}. To address this, we use the identity (see e.g.\ \cite[18.9.13]{NIST:DLMF}) $L^{(\alpha)}_{N-1}(x) = L_{N}^{(\alpha)}(x)-L_{N}^{(\alpha-1)}(x)$ in the first term and $L^{(\alpha+1)}_{N-2}(x) = L_{N-1}^{(\alpha+1)}(x)-L_{N-1}^{(\alpha)}(x)$ in the second term. This yields a cancellation of terms and we arrive at
    \begin{equation}
		\rho^{(2)}_{N}(x) = x^{\alpha}e^{-x}c_{N,\alpha}\left(L^{(\alpha)}_{N}(x)L_{N-1}^{(\alpha)}(x) - L^{(\alpha+1)}_{N-1}(x)L^{(\alpha-1)}_{N}(x)\right). \label{luebeta2dens}
    \end{equation}
    We now apply Lemma \ref{lem:lag} to each Laguerre polynomial in \eqref{luebeta2dens}. Combining this with $c_{N,\alpha} \sim N^{1-\alpha}$ as $N \to \infty$, we obtain
	\begin{equation}
		\rho^{(2)}_{N}(x) = O((N/x)^{\frac{1}{2}}), \qquad N \to \infty
	\end{equation}
	uniformly on $M/N < x < \delta N$. Then
	\begin{equation}
		\int_{M/N}^{\delta N}x^{-s}\rho^{(2)}_{N}(x)dx = O(N^{s}M^{\frac{1}{2}-s}).
	\end{equation}
	
	For $\beta=4$, we proceed similarly. Using Lemma \ref{lem:lag} on the correction term in \eqref{bet4exact}, we have
	\begin{equation}
		\frac{(2N)!}{\Gamma(2N+\alpha)}x^{\frac{\alpha}{2}-1}e^{-\frac{x}{2}}L^{(\alpha-1)}_{2N}(x) = O(N^{-\frac{\alpha-1}{2}}(Nx)^{-\frac{1}{4}}x^{-\frac{1}{2}}). \label{bet4prefactors}
	\end{equation}
	We split the integral over $t$ in \eqref{bet4exact} on $t \in [0,\frac{M}{N}]$ and $t \in [\frac{M}{N},x]$. In the first integral $I_{1}$ we change variables $t \to t/N$ and apply \eqref{hilb}, which gives $I_{1} = O(N^{\frac{\alpha}{2}})$. On the second integral $I_{2}$ we apply \eqref{lag-main-bnd} which shows $I_{2} = O(N^{\frac{\alpha}{2}+\frac{1}{4}}x^{\frac{1}{4}})$ which has greater magnitude than $I_{1}$. Multiplying the bound on $I_{2}$ with the bound in \eqref{bet4prefactors}, we get $\rho^{(4)}_{N}(x) = O((N/x)^{\frac{1}{2}}).$ This is the same bound we obtained in the $\beta=2$ case and the same reasoning establishes \eqref{orderbound} for $\beta=4$.
	
	For $\beta=1$, we begin by assuming $\alpha>0$. Then the integral over $t$ in \eqref{bet1exact} is  
	\begin{equation}
		2\int_{0}^{x}t^{\alpha-1}e^{-t}L^{(2\alpha-1)}_{N}(2t)dt - b_{N,\alpha}, \quad b_{N,\alpha} := \int_{0}^{\infty}t^{\alpha-1}e^{-t}L^{(2\alpha-1)}_{N}(2t)dt. \label{bet12integrals}
	\end{equation}
	It follows from the exact formula in \cite[Eq. (7.131)]{Forr} that $b_{N,\alpha} = O(N^{\alpha-1})$ as $N \to \infty$. Then we can use Lemma \ref{lem:lag} to estimate the first integral in \eqref{bet12integrals} as $O(x^{\frac{1}{4}}N^{\alpha-\frac{3}{4}})$ and on the interval in question, this term is of greater order than $b_{N,\alpha}$. Again by Lemma \ref{lem:lag}, the pre-factors in \eqref{bet1exact} are $O(N^{\frac{3}{2}-\alpha}x^{-\frac{1}{2}}(Nx)^{-\frac{1}{4}})$. Taking the product of these two error terms and proceeding as with $\beta=4$, we conclude that $\rho^{(1)}_{N}(x) = O((N/x)^{\frac{1}{2}})$ and therefore we have \eqref{orderbound} for $\beta=1$. 
    
    Finally, we deal with the  case $\beta=1$ with $-\frac{1}{2} < \alpha \leq 0$. First we note that the formula in \cite[Eq. (7.131)]{Forr} remains analytic in this range and continues to satisfy $b_{N,\alpha} = O(N^{\alpha-1})$, so we will focus on the first integral in \eqref{bet12integrals}. When $\alpha=0$ we use the identity (see e.g.\ \cite[18.36.2]{NIST:DLMF}) $L^{(-1)}_{N}(x) = \frac{-x}{N}L^{(1)}_{N}(x)$ and follow the same steps as for $\alpha>0$, so let us assume $-\frac{1}{2} < \alpha < 0$. Such values of $\alpha$ pose two difficulties. First, the Laguerre polynomials we need have parameters outside the classical range and secondly \eqref{bet12integrals} becomes non-integrable. As mentioned earlier and as per the discussion in \cite{W99}, we resolve this by analytic continuation as follows. We divide the first integral in \eqref{bet12integrals} over the two intervals $[0,\epsilon_{N}]$ and $[\epsilon_{N},x]$, where we choose $\epsilon_{N} = e^{-N}$. Integrating by parts, we write the integral over $[0,\epsilon_{N}]$ as
    \begin{equation}
2\int_{0}^{\epsilon_{N}}t^{\alpha}e^{-t}L_{N}^{(2\alpha)}(2t)dt-\int_{0}^{\epsilon_{N}}t^{\alpha}e^{-t}L_{N}^{(2\alpha-1)}(2t)dt,
\end{equation}
which is now integrable for all $\alpha>-1$ and represents the desired analytic continuation. By Lemma \ref{lem:lag}, these terms are at most $O(\epsilon_{N}^{3/4}N^{\alpha-1/4})$. Multiplying by the pre-factors in \eqref{bet1exact} we get bounds of order $O(\epsilon_{N}^{3/4}Nx^{-3/4})$. The contribution of this term to the integration over $M/N < x < \delta N$ is exponentially small in $N$. On the remaining integral over $[\epsilon_{N},x]$, we apply Lemma \ref{lem:lag} and follow the same steps as for $\alpha>0$.
\end{proof}

\section{Proofs for general {\texorpdfstring{$\beta > 0$}{beta}}}
\label{sec:fdresults}In this section, we prove Corollary \ref{thm:momsbeta} and Theorem \ref{thm:lowtemp}. We begin by recalling the exact formula for \eqref{momsintro} obtained in \cite[Eq. 33]{FD16}. It is stated there for the Jacobi $\beta$-ensemble which includes an extra parameter $b$. The Laguerre ensemble \eqref{laguerrebeta} is obtained by taking a suitable limit as $b \to \infty$ and this results in the following.
\begin{theorem}[Fyodorov and Le Doussal \cite{FD16}]\label{thmfd}
We have the exact formula
\begin{equation}\label{eq:FLD}
\mathcal{M}_{N}^{(\beta)}(-k,\alpha) = \sum_{\eta, |\eta|=k}A^{(\beta,N)}_{\eta}\alpha^{-}_{\eta}
\end{equation}
where, setting $\kappa = \frac{\beta}{2}$, we have
\begin{equation}\label{e:A}
\begin{split}
A^{(\beta,N)}_{\eta} &:= \frac{k(\eta_{1}-1)!} {(\kappa(\ell(\eta)-1)+1)_{\eta_1}} \prod_{i=2}^{\ell(\eta)} \frac{(\kappa(1-i))_{\eta_{i}}} {(\kappa(\ell(\eta)-i)+1)_{\eta_{i}}}
\prod_{i=1}^{\ell(\eta)}\frac{(\kappa(N-i+1))_{\eta_i}}{(\kappa(\ell(\eta)-i+1))_{\eta_i}}\\
&\times \prod_{1 \leq i < j \leq \ell(\eta)}\left(\frac{\kappa(j-i)+\eta_{i}-\eta_{j}}{\kappa(j-i)}\right)\frac{(\kappa(j-i+1))_{\eta_{i}-\eta_{j}}}{(\kappa(j-i-1)+1)_{\eta_{i}-\eta_{j}}}
\end{split}
\end{equation}
and $\alpha^{-}_{\eta}$ is given by \eqref{e:a2}.
\end{theorem}
With the help of computer algebra, it is easy to tabulate the first few moments explicitly using the above result.
\begin{align}
\mathcal{M}^{(\beta)}_{N}(-1,\alpha) &= \frac{N}{\alpha}\\
\mathcal{M}^{(\beta)}_{N}(-2,\alpha) &= \frac{N(\beta N+2\alpha)}{\alpha(\alpha-1)(\beta+2\alpha)}\\
\mathcal{M}^{(\beta)}_{N}(-3,\alpha) &= \frac{N(\beta N+\alpha)(\beta N+2\alpha)}{\alpha(\alpha-1)(\alpha-2)(\beta+\alpha)(\beta+2\alpha)}\\
\mathcal{M}^{(\beta)}_{N}(-4,\alpha) &= \frac{p_{4}}{\alpha(\alpha-1)(\alpha-2)(\alpha-3)(2\alpha-2+\beta)(\beta+\alpha)(\beta+2\alpha)(3\beta+2\alpha)}
\end{align}
 where
 \begin{equation} 	
 \begin{split}
 p_{4} &= N(\beta N+2\alpha)\left(N^{2}(5\alpha\beta^{2}+3\beta^{3}-6\beta^{2})\right.\\
 &\left.+N(10\alpha^{2}\beta+6\alpha\beta^{2}-12\alpha\beta)+4\alpha^{3}+2\alpha^{2}\beta-4\alpha^{2}+2\alpha\beta\right).
 \end{split}
 \end{equation}
\begin{proof}[Proof of Corollary \ref{thm:momsbeta}]
In \eqref{eq:FLD} for fixed $k$, the sum is over the finite set of partitions of $k$, so it is enough to pass to the limit term-by-term. For a fixed partition $\eta$, the only $N$-dependent factor is
\[
\prod_{i=1}^{\ell(\eta)}(\kappa(N-i+1))_{\eta_i}.
\]
For each $i=1,\ldots,\ell(\eta)$,
\[
\frac{(\kappa(N-i+1))_{\eta_i}}{N^{\eta_i}}
=
\kappa^{\eta_i}
\prod_{m=0}^{\eta_i-1}\left(1+\frac{m-\kappa(i-1)}{\kappa N}\right)
\longrightarrow
\kappa^{\eta_i},
\qquad N \to \infty.
\]
Since $|\eta|=\sum_{i=1}^{\ell(\eta)}\eta_i = k$, multiplying over $i$ gives
\[
N^{-k}\prod_{i=1}^{\ell(\eta)}(\kappa(N-i+1))_{\eta_i}
\longrightarrow
\kappa^{k},
\qquad N \to \infty.
\]
All remaining factors in \eqref{e:A} are independent of $N$, hence
\[
\lim_{N \to \infty} N^{-k}A^{(\beta,N)}_{\eta}
=
A^{(\beta)}_{\eta},
\]
where $A^{(\beta)}_{\eta}$ is given by \eqref{AlargeN}. Passing to the limit term-by-term, we obtain
 \eqref{betform}.
\end{proof}

\begin{corollary}
\label{cor:largebeta}
The following large $\beta$ limit of the moments defined by
\begin{equation}
\mathcal{M}(-k,\nu) := \lim_{\beta \to \infty}\beta^{k}\mathcal{M}^{(\beta)}\left(-k,\frac{\beta}{2}(\nu+1)\right) 
\end{equation}
exists for any $\nu>-1$. We have the explicit formula
\begin{equation}
\mathcal{M}(-k,\nu) = \sum_{\eta, |\eta|=k}A_{\eta}\left(\prod_{i=1}^{\ell(\eta)}(i+\nu)^{-\eta_{i}}\right), \label{firstNthenbet}
\end{equation}
where
\begin{equation}\label{e:A2}
\begin{split}
A_{\eta} &:= \frac{2^k k(\eta_{1}-1)!} {\eta_{\ell(\eta)}!}(-1)^{k-\eta_{1}} \prod_{i=2}^{\ell(\eta)}(i-1)^{\eta_{i}}\prod_{i=1}^{\ell(\eta)-1}(\ell(\eta)-i)^{-\eta_{i}}\\
&\times \prod_{i=1}^{\ell(\eta)}(\ell(\eta)-i+1)^{-\eta_{i}}\prod_{i=2}^{\ell(\eta)}\frac{1}{(\eta_{i-1}-\eta_{i})!}\\
&\times \prod_{1\leq p < q \leq \ell(\eta)}(q-p+1)^{\eta_{p}-\eta_{q}}\prod_{q=3}^{\ell(\eta)}\prod_{p=1}^{q-2}(q-p-1)^{\eta_{q}-\eta_{p}}.
\end{split}
\end{equation}
Furthermore, this result is independent of the order of the two limits in $N$ and $\beta$. Namely, if we instead consider the limit with \textit{fixed $N$}, $a = \frac{\beta}{2}(\nu+1)$ and $\beta \to \infty$, followed by $N \to \infty$, we obtain the same result. That is, we have
\begin{equation}
\lim_{N \to \infty}\lim_{\beta \to \infty}\left(\frac{\beta}{N}\right)^{k}\mathcal{M}^{(\beta)}_{N}\left(-k,\frac{\beta}{2}(\nu+1)\right) = \mathcal{M}(-k,\nu). \label{firstbetthenN}
\end{equation}
If $\beta := \beta_{N}$ is any sequence such that $\beta_{N} \to \infty$ as $N \to \infty$, we again have
\begin{equation}
\lim_{N \to \infty}\left(\frac{\beta_{N}}{N}\right)^{k}\mathcal{M}^{(\beta_{N})}_{N}\left(-k,\frac{\beta_{N}}{2}(\nu+1)\right) = \mathcal{M}(-k,\nu). \label{oneNlimit}
\end{equation}
\end{corollary}

\begin{proof}
Recall that $\kappa=\beta/2$. Since the set of partitions of $k$ is finite, all limits may be taken term by term. Fix a partition $\eta$ of $k$ and set $\ell=\ell(\eta)$. We shall use the elementary large-$\kappa$ asymptotics
\[
(\kappa c+d)_{m} = (\kappa c)^{m}\left(1+O(\kappa^{-1})\right), \qquad
(\kappa c+d)_{-m} = (\kappa c)^{-m}\left(1+O(\kappa^{-1})\right),
\]
valid for fixed $m\in \mathbb{N}$, fixed $d$ and $c>0$.

To prove \eqref{firstNthenbet}, we note that the factor \eqref{e:a2} satisfies
\[
\alpha^{-}_{\eta}
=
\prod_{i=1}^{\ell}(\alpha+1+\kappa(i-1))_{-\eta_{i}}
=
\kappa^{-k}\left(\prod_{i=1}^{\ell}(i+\nu)^{-\eta_{i}}\right)\left(1+O(\kappa^{-1})\right).
\]
Now consider the coefficient \eqref{AlargeN}. If $\ell=1$, then $\eta=(k)$ and a direct computation from \eqref{betform} gives
\[
\beta^{k}A^{(\beta)}_{(k)}\alpha^{-}_{(k)} \to 2^{k}(1+\nu)^{-k},
\]
which agrees with \eqref{e:A2}. Otherwise if $\ell\ge 2$, the individual factors in \eqref{AlargeN} have the following asymptotics
\begin{align*}
\frac{1}{(\kappa(\ell-1)+1)_{\eta_{1}}}
&=
\kappa^{-\eta_{1}}(\ell-1)^{-\eta_{1}}\left(1+O(\kappa^{-1})\right), \\
\frac{(\kappa(1-i))_{\eta_{i}}}{(\kappa(\ell-i)+1)_{\eta_{i}}}
&=
(-1)^{\eta_{i}}(i-1)^{\eta_{i}}(\ell-i)^{-\eta_{i}}\left(1+O(\kappa^{-1})\right),
\qquad 2\le i\le \ell-1, \\
\frac{(\kappa(1-\ell))_{\eta_{\ell}}}{(\kappa(\ell-\ell)+1)_{\eta_{\ell}}}
&=
(-1)^{\eta_{\ell}}\frac{\kappa^{\eta_{\ell}}(\ell-1)^{\eta_{\ell}}}{\eta_{\ell}!}
\left(1+O(\kappa^{-1})\right), \\
\frac{1}{(\kappa(\ell-i+1))_{\eta_{i}}}
&=
\kappa^{-\eta_{i}}(\ell-i+1)^{-\eta_{i}}\left(1+O(\kappa^{-1})\right),
\qquad 1\le i\le \ell .
\end{align*}
For the pairwise factors, writing $d_{pq}:=\eta_{p}-\eta_{q}$, we have
\[
\frac{\kappa(q-p)+d_{pq}}{\kappa(q-p)} = 1+O(\kappa^{-1}),
\]
and
\[
\frac{(\kappa(q-p+1))_{d_{pq}}}{(\kappa(q-p-1)+1)_{d_{pq}}}
=
\begin{cases}
\displaystyle
\kappa^{d_{pq}}\frac{(q-p+1)^{d_{pq}}}{d_{pq}!}\left(1+O(\kappa^{-1})\right),
& q=p+1,\\[1ex]
\displaystyle
\left(\frac{q-p+1}{q-p-1}\right)^{d_{pq}}\left(1+O(\kappa^{-1})\right),
& q\ge p+2.
\end{cases}
\]
Multiplying these asymptotics and using $\beta^{k}=(2\kappa)^{k}$, all powers of $\kappa$ cancel exactly. The $\eta$-th summand in \eqref{betform} therefore converges to
\[
A_{\eta}\prod_{i=1}^{\ell}(i+\nu)^{-\eta_{i}},
\]
with $A_{\eta}$ given by \eqref{e:A2}. Summing over the finitely many partitions of $k$ yields \eqref{firstNthenbet}.

To prove \eqref{firstbetthenN}, start from the finite-$N$ formula in Theorem \ref{thmfd}. The only new factor relative to the previous computation is
\[
\frac{(\kappa(N-i+1))_{\eta_{i}}}{(\kappa(\ell-i+1))_{\eta_{i}}}
=
N^{\eta_{i}}
\left(1-\frac{i-1}{N}\right)^{\eta_{i}}
(\ell-i+1)^{-\eta_{i}}
\left(1+O(\kappa^{-1})\right).
\]
All other factors are exactly the same as above. Hence, for each fixed $\eta$,
\[
\lim_{\beta\to\infty}
\left(\frac{\beta}{N}\right)^{k}A^{(\beta,N)}_{\eta}\alpha^{-}_{\eta}
=
A_{\eta}
\left(\prod_{i=1}^{\ell(\eta)}\left(1-\frac{i-1}{N}\right)^{\eta_{i}}\right)
\left(\prod_{i=1}^{\ell(\eta)}(i+\nu)^{-\eta_{i}}\right).
\]
Summing over $\eta$ gives
\[
\lim_{\beta\to\infty}\left(\frac{\beta}{N}\right)^{k}\mathcal{M}_{N}^{(\beta)}
\left(-k,\frac{\beta}{2}(\nu+1)\right)
=
\sum_{\eta, |\eta|=k}
A_{\eta}
\left(\prod_{i=1}^{\ell(\eta)}\left(1-\frac{i-1}{N}\right)^{\eta_{i}}\right)
\left(\prod_{i=1}^{\ell(\eta)}(i+\nu)^{-\eta_{i}}\right).
\]
Since the set of partitions of $k$ is finite, letting $N\to\infty$ and using $\left(1-\frac{i-1}{N}\right)^{\eta_i}\to 1$ for each fixed partition $\eta$, we obtain \eqref{firstbetthenN}.

Finally, if $\beta=\beta_{N}$ with $\beta_{N}\to\infty$, the same asymptotics with $\kappa=\beta_{N}/2$ give, for each fixed partition $\eta$,
\[
\left(\frac{\beta_{N}}{N}\right)^{k}A^{(\beta_{N},N)}_{\eta}a^{-}_{\eta}
=
A_{\eta}
\left(\prod_{i=1}^{\ell(\eta)}\left(1-\frac{i-1}{N}\right)^{\eta_{i}}\right)
\left(\prod_{i=1}^{\ell(\eta)}(i+\nu)^{-\eta_{i}}\right)
+o(1),
\qquad N\to\infty.
\]
Summing over the finitely many partitions of $k$ and letting $N\to\infty$ proves \eqref{oneNlimit}.
\end{proof}

We now show that the limiting expression $\mathcal{M}(-k,\nu)$ in \eqref{firstNthenbet} coincides with the even integer values of the Bessel zeta function:
\begin{equation}
\mathcal{M}(-k,\nu) = 8^{k}\zeta_{\nu}(2k), \qquad k \in \mathbb{N}.
\end{equation}
\begin{proof}[Proof of Theorem \ref{thm:lowtemp}]
Since Corollary \ref{cor:largebeta} shows that the relevant limiting operations can be interchanged, it suffices to compute just one of the limits. In particular, in the following we will compute the limit \eqref{firstbetthenN} using an alternative approach based on the ideas in \cite{DE05}. To introduce this, recall that a random variable $\chi_{c}$ is said to have the $\chi$ distribution with parameter $c>0$ if it has probability density function
\begin{equation}
f_{\chi_c}(x) = \frac{x^{c-1}e^{-x^{2}/2}}{2^{c/2-1}\Gamma\left(\frac{c}{2}\right)}, \qquad x > 0.
\end{equation}
Then \cite{DE05} shows that the points $\{\lambda_{j}\}_{j=1}^{N}$ of the Laguerre $\beta$-ensemble are distributed according to the eigenvalues of the matrix $L_{\beta} := B_{\beta}B_{\beta}^{\mathrm{T}}$ where $B_{\beta}$ is the following $N \times N$ lower bi-diagonal matrix:
\begin{equation}
B_{\beta} = \frac{1}{\sqrt{2}}\begin{pmatrix}
\chi_{2\tilde{a}}      & 0        & 0        & \cdots & 0        \\
\chi_{\beta(N-1)}   & \chi_{2\tilde{a}-\beta}      & 0        & \cdots & 0        \\
0        & \chi_{\beta(N-2)}   & \chi_{2\tilde{a}-2\beta}      & \cdots & 0        \\
\vdots   & \vdots   & \vdots   & \ddots & \vdots   \\
0        & 0        & 0        & \chi_{\beta} & \chi_{2\tilde{a}-\beta(N-1)}
\end{pmatrix},
\end{equation}
where all the $\chi$ random variables are independent, and
\begin{equation}
\tilde{a} := \alpha+\frac{\beta}{2}(N-1)+1.
\end{equation}
In the scaling regime of interest $\alpha = \frac{\beta}{2}(\nu+1)$ and hence
\begin{equation}
\tilde{a} = \frac{\beta}{2}(N+\nu)+1.
\end{equation}
Then in the limit $\beta \to \infty$, all $\chi$ variables have large parameter. It follows from the strong law of large numbers that we have the almost sure convergence $\lim_{\beta \to \infty}\sqrt{2/\beta}B_{\beta} = B$ where
\begin{equation}
B =  \begin{pmatrix}
\sqrt{N+\nu}     & 0        & 0        & \cdots & 0        \\
\sqrt{N-1}   & \sqrt{N+\nu-1}      & 0        & \cdots & 0        \\
0        & \sqrt{N-2}   & \sqrt{N+\nu-2}      & \cdots & 0        \\
\vdots   & \vdots   & \vdots   & \ddots & \vdots   \\
0        & 0        & 0        & \sqrt{1} & \sqrt{\nu+1}
\end{pmatrix}.
\end{equation}
The eigenvalues of the limiting matrix $L = BB^{\mathrm{T}}$ are precisely the zeros of the Laguerre polynomial $L_{N}^{(\nu)}(x)$ of parameter $\nu$ and degree $N$, see \cite{DE05} for further details. We denote these zeros $\{l^{(N)}_{n}(\nu)\}_{n=1}^{N}$ and arrange them in increasing order $l^{(N)}_{1}(\nu) < l^{(N)}_{2}(\nu) < \ldots < l^{(N)}_{N}(\nu)$. As this limit is for each fixed size matrix $N$, it also holds in the sense of joint expectations of finite collections of the independent $\chi$ random variables, i.e.\ for any multivariate Laurent polynomial of fixed number of variables $N$ and fixed degree, we have
\begin{equation}
\lim_{\beta \to \infty}\mathbb{E}(f(\tilde{\chi}_{a_{1},\beta},\ldots,\tilde{\chi}_{a_N, \beta})) = f(\sqrt{a_1},\ldots,\sqrt{a_N}),
\end{equation}
where $\tilde{\chi}_{c,\beta} = \frac{\chi_{c\beta}}{\sqrt{\beta}}$ and $a_{1},\ldots,a_{N}$ are positive parameters. We conclude that
\begin{equation}
\lim_{\beta \to \infty}\beta^{k}\mathbb{E}\left(\sum_{n=1}^{N}\lambda_{n}^{-k}\right) = 2^{k}\sum_{n=1}^{N}\left(l^{(N)}_{n}(\nu)\right)^{-k}. \label{sumoflaguerre}
\end{equation}
This deals with the first limit as $\beta \to \infty$ in \eqref{firstbetthenN}. To compute the subsequent limit $N \to \infty$, note that for a fixed $n$ in the summand on the right-hand side of \eqref{sumoflaguerre}, we have the pointwise limit
\begin{equation}
\lim_{N \to \infty}4Nl^{(N)}_{n}(\nu) = j^{2}_{\nu,n} ,
\end{equation}
where $\{j_{\nu,n}\}_{n=1}^{\infty}$ are the zeros of the Bessel function $J_{\nu}(x)$, arranged in increasing order, see  \cite[Section 18.16]{NIST:DLMF}. To compute the limit of the sum in \eqref{sumoflaguerre} we apply dominated convergence. We have the following uniform lower bound \cite[18.16.10]{NIST:DLMF}
\begin{equation}
l^{(N)}_{n}(\nu) > \frac{j^{2}_{\nu,n}}{4N+2\nu+2}, \qquad n=1,\ldots,N.
\end{equation} 
This implies
\begin{equation}
\sum_{n=1}^{N}\left(l^{(N)}_{n}(\nu)\right)^{-k} \leq C_{k}(\nu)N^{k}\sum_{n=1}^{\infty}j^{-2k}_{\nu,n}
\end{equation}
where $C_{k}(\nu)$ is a positive constant depending only on $k$ and $\nu$. We conclude that
\begin{equation}
\lim_{N \to \infty}\lim_{\beta \to \infty}\left(\frac{\beta}{N}\right)^{k}\mathbb{E}\left(\sum_{n=1}^{N}\lambda_{n}^{-k}\right) = 8^{k}\sum_{n=1}^{\infty}j^{-2k}_{\nu,n} = 8^{k}\zeta_{\nu}(2k).
\end{equation}
\end{proof}

\bibliographystyle{plain}
\bibliography{hardedgebib}

\end{document}